\newcommand{\D}{\cal{D}}

\newtheorem{lem}  {Lemma} 
\newcommand {\BL} {\begin{lem}} 
\newcommand {\EL} {\end{lem}} 

\newtheorem{cor}  {Corollary} 
\newcommand {\BCR} {\begin{cor}}
\newcommand {\ECR} {\end{cor}}

\newtheorem{thm} {Theorem} 
\newcommand {\BT} {\begin{thm}}
\newcommand {\ET} {\end{thm}}

\newtheorem{defi} {Problem} 

\newcommand {\BDE} {\begin{defi}}
\newcommand {\EDE} {\end{defi}}
\newcommand{\comment}[1]{}

\newtheorem{xxx} {Definition} 
\newcommand {\BD} {\begin{xxx}}
\newcommand {\ED} {\end{xxx}}

\newtheorem{xxxxx} {Observation} 
\newcommand {\BO} {\begin{xxxxx}}
\newcommand {\EO} {\end{xxxxx}}

\documentclass[11pt]{llncs}
\usepackage{amsmath}
\usepackage{graphicx}
\usepackage{algorithm}
\usepackage{algorithmic}

\long\def\symbolfootnote[#1]#2{\begingroup%
\def\thefootnote{\fnsymbol{footnote}}\footnote[#1]{#2}\endgroup}

\begin{document}

\title{Towards an Optimal Space-and-Query-Time Index for \\ Top-$k$ Document Retrieval}
\author{Wing-Kai Hon\inst{1}, Rahul Shah\inst{2}, and
        Sharma V. Thankachan\inst{2}  }
\institute{
     Department of CS, National Tsing Hua University,
     Taiwan. \email{wkhon@cs.nthu.edu.tw}
\and
     Department of CS, Louisiana State University, USA.
     \email{\{rahul,thanks\}@csc.lsu.edu}
}
\date{}

\maketitle

\begin{abstract}

Let $\D = $$ \{d_1,d_2,...d_D\}$ be a given set of  $D$ string documents of total length $n$, our task is to index $\D$, such that the $k$ most relevant documents for an online query pattern $P$ of length $p$ can be retrieved efficiently. We propose an index of size $|CSA|+n\log D(2+o(1))$ bits and $O(t_{s}(p)+k\log\log n+poly\log\log n)$ query time for the basic relevance metric \emph{term-frequency}, where  $|CSA|$ is the size (in bits) of a compressed full text index of $\D$, with $O(t_s(p))$ time for searching a pattern of length $p$ . We further reduce the space to $|CSA|+n\log D(1+o(1))$ bits, however the query time will be 
$O(t_s(p)+k(\log \sigma \log\log n)^{1+\epsilon}+poly\log\log n)$, where $\sigma$ is the alphabet size and $\epsilon >0$ is any constant.

\end{abstract}

\section{Introduction and Related Work}

Document retrieval is a special type of pattern matching that is closely related to information retrieval and web searching.  In this problem, the data consists of a collection of text documents, and given a query pattern $P$, we are required to report all the documents in which this pattern occurs (not all the occurrences). In addition, the notion of \emph{relevance} is commonly applied to rank all the documents that satisfy the query, and only those documents with the highest relevance are returned. Such a concept of relevance has been central in the effectiveness and usability of present day search engines like Google, Bing, Yahoo, or Ask. When relevance is considered, the query has an additional input parameter $k$, and the task is to report the $k$ documents with the highest relevance to the query pattern (in the decreasing order of relevance), instead of finding all the documents that contain the query pattern (as there may be too many).
More formally, let $\cal{D} = $$ \{d_1,d_2,...d_D\}$ denote a given set of $D$ string documents to be indexed, whose total lengths is $n$,
and let $P$ denote a query pattern of length $p$.  Let $occ$ be the number of occurrences of this pattern over the entire collection $\cal{D}$, and $ndoc$ be the number of documents out of $D$ in which the pattern $P$ appears.  One of the main issues is the fact that $k \ll ndoc \ll occ$. Thus, it is important to design indexes which do not have to go through all the occurrences or even all the documents in order to answer a query.

The research in string document retrieval was introduced by Matias et al.~\cite{matias}, and Muthukrishnan~\cite{muthu} formalized it with the introduction of relevance metrics like \emph{term-frequency (tf)} and \emph{min-dist},\footnote{$tf(P,d)$ is the number of occurrences of $P$ in $d$ and \emph{min-dist$(P,d)$} is the minimum distance between two occurrences of $P$ in $d$} and proposed indexes with efficient query performance.
Since then, this has been an active research area~\cite{flexible,VM}.
The top-$k$ document retrieval problem was introduced in~\cite{JDA}, where an $O(n\log n)$-word index is proposed with $O(p+k+\log n\log\log n)$ query time for the case when the relevance metric is \emph{term-frequency}.
A recent flurry of activities in this area~\cite{xxx,CPM10,gagie,Belazz,culpepper,sigir,SEA,Nekrich,zzz,SEA12} came with Hon et al.'s work~\cite{FOCS09} where they gave a linear-space index with $O(p+k\log k)$ query time, which works for a wide class of relevance metrics.
The recent structure by Navarro and Nekrich~\cite{soda12} achieves optimal $O(p+k)$ query time using $O(n(\log\sigma+\log D+\log\log n))$ bits, which improves the results in~\cite{FOCS09} in both space and time. If the relevance metric is \emph{term-frequency}, their index space can be further improved to $O(n(\log\sigma+\log D))$ bits. All these interesting results have contributed towards the goal of achieving an optimal query time index. However, the space is far from optimal, moreover the constants hidden in the space bound can restrict the use of these indexes in practice. On the other side, the succinct index proposed by Hon et al.~\cite{FOCS09} takes about $O(\log^4 n)$ time to report each document, which is likely to be impractical. This time bound has been further improved by~\cite{Belazz,gagie}, but still ${\rm polylog}(n)$ time is required per reported document.
Another line of work is to derive indexes using about $n\log D$ bits additional space, and the best known index takes a per document report time of $O(\log k \log^{1+\epsilon} n)$~\cite{Belazz}. Efficient practical indexes are also known~\cite{culpepper}, but their query algorithms are heuristics with no worst-case bound. 
In this paper, we introduce two space efficient indexes with per document report time \emph{poly}-log-logarithmic in $n$. 
The main results are summarized as follows.

\BT
There exists an index of size $|CSA|+n\log D(2+o(1))$ bits with a query time of $O(t_s(p)+k\log\log n+ poly\log\log n)$ for retrieving top-$k$ documents with the highest term frequencies, where $|CSA|$ is the size (in bits) of a compressed full text index of $\D$ with $O(t_s(p))$ time for searching a pattern of length $p$.
\ET

\BT
There exists an index of size $|CSA|+n\log D(1+o(1))$ bits with a query time of $O(t_s(p)+k(\log \sigma \log\log n)^{1+\epsilon}+poly\log\log n)$ for retrieving top-$k$ documents with the highest term frequencies, where $|CSA|$ is the size (in bits) of a compressed full text index of $\D$ with $O(t_s(p))$ time for searching a pattern of length $p$, $\sigma$ is the alphabet size and $\epsilon >0$ is a constant.
\ET

Table~\ref{table1} gives a summary of the major results in the top-$k$ frequent document retrieval problem. 
The time complexities are simplified by assuming that we are using the full text index proposed by Belazzougui and Navarro, of size $|CSA|= nH_h+O(n)+o(n\log\sigma)$ bits and $t_s(p)=O(p)$, where $H_h$ is the $h$th order empirical entropy of $\D$~\cite{csanew}. We also assume $D < n^{\varepsilon}$ for some $\varepsilon <1$ and $\epsilon >0$ is any constant.

\begin{table}
\caption{Indexes for Top-$k$ Frequent Document Retrieval} 
\label{table1}
\begin{center}
\small
\begin{tabular}{| l | l | l | l } 
\hline
\textbf{Source} & \textbf{Index Space (in bits)} & \textbf{Time per reported document} \\
\hline \hline 
 ~\cite{JDA} & $O(n\log n +n\log^2 D)$ & $O(1)$  \\ \hline
 
 ~\cite{FOCS09} & $O(n\log n)$ &  $O(\log k)$ \\  \hline

~\cite{culpepper}  & $|CSA| + n\log D(1+o(1))$  & Unbounded\\  \hline 
 
~\cite{FOCS09}  & $2|CSA| +o(n)$  & $O(\log^{4+\epsilon} n) $\\  \hline 

~\cite{Belazz} & $2|CSA| +o(n)$  & $O(\log k \log^{2+\epsilon} n) $\\  \hline

~\cite{gagie} & $|CSA| +O(\frac{n\log D}{\log\log D})$ & $O(\log^{3+\epsilon} n) $ \\  \hline 

~\cite{Belazz} & $|CSA| +O(\frac{n\log D}{\log\log D})$  & $O(\log k \log^{2+\epsilon} n) $\\  \hline 

~\cite{Belazz} & $|CSA| +O(n\log\log\log D)$  & $O(\log k \log^{2+\epsilon} n) $\\  \hline 

 ~\cite{soda12} & $O(n\log\sigma+n\log D)$ & $O(1)$ \\ \hline

~\cite{gagie} & $|CSA| +n\log D+o(n)$  & $O(\log^{2+\epsilon} n) $ \\  \hline

~\cite{Belazz} & $|CSA| +n\log D+o(n)$  & $O(\log k \log^{1+\epsilon} n) $\\  \hline

\hline \hline 
Ours & $|CSA| +2n\log D(1+o(1))$  & $O(\log\log n) $\\  \hline 

Ours & $|CSA| +n\log D(1+o(1))$  & $O((\log \sigma \log\log n)^{1+\epsilon}) $\\
 \hline \end{tabular}
\end{center}       
\end{table}

\section{Preliminaries}

\subsection{Top-$k$ Using Range Maximum/Minimum Queries}
One of the main tools in top-$k$ retrieval is the \emph{range maximum/minimum query structures} (RMQ)~\cite{rmq}.
We summarize the results in the following lemmas (We defer the proofs to the Appendix A and B respectively).

\begin{lemma}
Let $A[1...n]$ be an array of $n$ numbers. We can preprocess $A$ in linear time and associate $A$ with a $2n+o(n)$ bits RMQ data structure such that given a set of $t$ non-overlapping ranges $[L_1, R_1], [L_2, R_2],\ldots,[L_t, R_t]$, we can find the largest (or smallest) $k$ numbers in $A[L_1..R_1] \cup A[L_2..R_2]\cup\cdots\cup A[L_t..R_t]$ in unsorted order in $O(t+k)$ time.
\end{lemma}

\begin{lemma}
Let $A[1...n]$ be an array of $n$ integers taken from the set $[1,\pi]$, and each number $A[i]$ is associated with a score (which may be stored separately and can be computed in $t_{score}$ time). Then the array $A$ can be maintained in $O(n\log \pi)$ bits, such that given two ranges $[x',x'']$, $[y', y'']$, and a parameter $k$, we can search among those entries $A[i]$ with $x' \leq i \leq x''$ and $y' \leq A[i] \leq y''$, and report the $k$ highest scoring entries in unsorted order in $O((\log \pi + k)(\log \pi+t_{score}))$ time.
\end{lemma}

\section{A Brief Review of Hon et al.'s  Index}
In this section we give a brief description of Hon et al.'s index~\cite{FOCS09}. 
Let $T=d_1 \# d_2 \# \cdots \# d_D \#$ be a text obtained by concatenating all the documents in $\cal{D}$, separated by a special symbol $\#$ not appearing elsewhere inside any of the $d_i$s.  Then the suffix tree~\cite{Wei73,McC76,sa} of $T$ is called the \emph{generalized suffix tree} GST of $\cal{D}$.  Then any given substring $T[a...b]$ (which does not contain $\#$) of $T$ is a substring of some document $d_x \in \cal{D}$, and the value of $x$ can be computed in $O(1)$ time by maintaining an $(n+D)(1+o(1))$-bit auxiliary data structure\footnote{Maintain a bit vector $\cal{B}$$[1...(n+D)]$, where $\cal{B}$$[i]=1$ if and only if $T[i]=\#$, then $x= rank_{\cal{B}}(a)+1$ and can be computed in $O(1)$ time using~\cite{RRR}.}.
Each edge in GST is labeled by a character string
and for any node $u$, the \emph{path label} of $u$, denoted by $path(u)$
is the string formed by concatenating the edge labels from root to $u$. Note that the path label of the $i$th leftmost leaf in GST
is exactly the $i$th lexicographically smallest suffix of $T$.
For a pattern $P[1..p]$ that appears in $T$, the \emph{locus node} of $P$ is denoted by $locus(P)$, which is the unique node closest to the root
such that $P$ is a prefix of $path(locus(P))$, and can be determined in $O(p)$ time.
We augment the following structures on GST.

\emph{N-structure}: An N-structure entry is a triplet  $(doc, score, parent)$ and is associated with some node in GST. 
If $u$ is a leaf node with $path(u)$ is a suffix of document $d$, the an N-structure entry with $doc=d$ is stored at $u$. 
However, if it is an internal node, multiple N-structure entries may be stored at $u$ as follows: an entry with $doc = d$ is stored if and only if at least two children of $u$ contain (a suffix of) document $d$ in their subtrees. The $score$ field in an N-structure entry for a document $d$ associated with a node $u$ is $score(path(u),d)$: the relevance score of $d$ with respect to the pattern $path(u)$\footnote{The $score$ is dependent only on $d$ and the set of occurrences of $path(u)$ in $d$.}.
The \emph{parent} field stores (the pre-order rank of) the lowest ancestor of $u$
which has an entry for document $d$ in its N-structure. In case there is no such ancestor, we assign a dummy node which is regarded as the parent of the root of GST.

\emph{I-structure}: An I-structure entry is a triplet  $(doc, score, origin)$ and is associated with some node in GST.
If node $u$ has an N-structure entry for document $d$ and an N-structure entry of another node $v$ is given by $(d, score(path(v),d), u)$, then $u$ will have an I-structure entry  $(d, score(path(v),d), v)$.  An internal node may be associated with multiple I-structure entries, and these entries are maintained in an array,  sorted by
the \emph{origin} field. In addition, a range maximum query (RMQ) structure is maintained over the array based on the \emph{score} field.

\subsection{Query Answering}
To answer a top-$k$ query, we first search for the query pattern $P$ in GST and find its locus node $locus(P)$. We also find the rightmost leaf $locus_R(P)$ in the subtree of $locus(P)$.  Now, our task is to find, among the documents whose suffixes appear in the subtree of $locus(P)$, which $k$ of them have the highest
occurrences of $P$.  Hon et al.\ showed that this can be done by checking only the I-structure entries associated with the proper ancestors of $locus(P)$,
and then retrieving those $k$ entries which has the highest  \emph{score} values and whose \emph{origin} is from the subtree of $locus(P)$ (inclusively).
The number of ancestors of $P$ is bounded by $p$ and since the I-structure entries are sorted according to the origin values, the entries to be checked will occupy a contiguous region in the sorted array.  The boundaries of the contiguous region can be obtained by performing a binary search based on
(the pre-order ranks of) $locus(P)$ and $locus_R(P)$. Once we get the boundaries of the contiguous region in each proper ancestors of $locus(P)$,
we can apply RMQ queries repeatedly over \emph{score} and retrieve
the top-$k$ scoring documents in sorted order in $O(p\log n+k\log k)$ time.
The binary search step can be made faster by maintaining a predecessor structure~\cite{yfast} and the resulting time will become $O(p\log\log n +k\log k)$. This time has been further improved to $O(p+k\log k)$ by introducing two additional fields $\delta_f$ and $\delta_{\ell}$ in each N-structure entry. The number of N-structure entries (hence I-structure entries) is $\leq 2n$. Therefore the index space is $O(n\log n)$ bits.

\section{Our Linear-Space Index}

In this section, we derive a modified version of Hon et al.'s linear index without $\delta$ fields and still achieve $O(p)$ term in query time. The main technique is by introducing a novel criterion that categorizes the I-structure entries as \emph{near} and \emph{far}.  The \emph{far} entries associated with certain nodes can be maintained together as a combined I-structure, which reduces the number of I-structure boundaries to be searched to $O(p/\pi+\pi)$, where $\pi$ is a sampling factor. By choosing $\pi=\log\log n$, we shall use predecessor search structure (instead of $\delta$ fields) and can compute the I-structure boundaries in $O((p/\pi+\pi)\log\log n)=O(p+\log^2\log n)$ time. We have the following result.
\BT
There exists an index of size $O(n\log n)$ bits for top-$k$ document retrieval with $O(p+\log^2\log n+ k\log\log\log n+ k\log k)$ query time.
\ET

\begin{proof}
Firstly, we mark all nodes in GST whose node-depths are multiples of $\pi$ (node-depth of root is $0$).  Thus, any unmarked node is at most $\pi$ nodes away
from its lowest marked ancestor. Also, the number of marked ancestors of any node $= \lceil$(number of ancestors)$/\pi \rceil$.
For any node $w$ in GST, we define a value $\zeta(w) <\pi$, where $\zeta(w) =0$ if $w$ is marked, else it is the number of nodes in the path from $w$ (exclusively) till its lowest marked ancestor (inclusively).
In each I-structure entry $(d, s, v)$ associated with a node $w$, we maintain a fourth component $\zeta(w)$.
Next, we categorize the I-structure entries as \emph{far} and \emph{near} as follows:
\begin{quote}
\emph{An I-structure entry associated with a node $w$,  with $origin = v$, is near if there exists no marked node in the path from $v$ (inclusively) to $w$ (exclusively), else it is far.}
\end{quote}
We restructure the entries such that all \emph{far} entries are maintained in a combined I-structure associated with some marked nodes as follows: if $(d, s, v,\zeta(w))$ is a \emph{far} entry in the I-structure $I_w$ associated with node $w$, then we remove this entry from $I_w$ and move to a combined I-structure associated with the node $u$, where $u= w$ if $w$ is  marked, else $u$ is the lowest marked ancestor of $w$ (i.e.,  $u$ is $\zeta(w)$ nodes above $w$). All the entries in the combined I-structure are maintained in the sorted order of \emph{origin} values.  A predecessor search structure over the \emph{origin} field and RMQ structure over the \emph{score} field is maintained over all I-structures.
Next, to understand how to answer a query with our index, we introduce the following auxiliary lemma.
\begin{lemma}
The top-$k$ documents corresponding to a pattern $P$ can be obtained by checking the following I-structure entries $($with origins coming from the subtree of $locus(P))$:\\
(i) near entries in the regular I-structures associated with the  nodes in the path from $locus(P)$ (exclusively) till its lowest marked ancestor $u$ (inclusively), and there are at most $\pi$ such nodes; \\
(ii) far entries with $\zeta < \zeta(locus(P))$ in the combined I-structure of $u$, and\\
(iii) far entries in the combined I-structures associated with the marked proper (at most $p/\pi$) ancestors of $u$.
\end{lemma}
\begin{proof}
In the original index by Hon et al.,  we need to check the I-structure entries in all ancestors of $locus(P)$. We may categorize them as follows:
\begin{itemize}
\item[(a)] \emph{near} entries  associated with a node in the subtree of $u$ (inclusively);
\item[(b)] \emph{far} entries associated with a node in the subtree of $u$ (inclusively);
\item[(c)] \emph{far} entries  associated with an ancestor node of $u$;
\item[(d)] \emph{near} entries  associated with an ancestor node of $u$.
\end{itemize}

All entries in (a) belong to category (i) in the lemma. The valid entries in (b) belong to
category (ii), where the inequality $\zeta < \zeta(locus(P))$ ensures that the all entries in category (ii) were originally from an ancestor of $locus(P)$ .
All those entries in (c), which may be a possible candidate for the top-$k$ documents, belong to category (iii) in the lemma. None of the entries in (d) can be a valid output, as the origin of those entries are not coming from the subtree of $u$ (from the definition of  a \emph{near} entry), hence not from the subtree of $locus(P)$.
On the other hand, since we always check for the entries with origins coming from the subtree of $locus(P)$, these entries must be a subset of those checked in the original index by Hon et al.  In conclusion, the entries checked in both indexes are exactly the same, and the lemma follows.  \qed
\end{proof}

Based on the above lemma, we may compute $k$ candidate answers from each category and the actual top-$k$ answers can be computed by comparing the score of these $3k$ documents.
In category (i) we have at most $\pi$ boundaries to be searched, which takes $O(\pi\log\log n)$ time, and then retrieve the $k$ candidate answers in the unsorted order in $O(\pi+k)$ time using lemma 1. Similarly in category (iii), the number of I-structure boundaries to be searched is $p/\pi$ and it takes total $O((p/\pi)\log\log n+k)$ time. However, for category (ii), we have an additional constraint on $\zeta$ value of the entries.  To facilitate the process, the $\zeta$ components are maintained by the data structure in Lemma 2 in $O(n\log\pi)$ bits, so that the desired answers can be reported in $O((\log \pi +k)(\log \pi + O(1)))$ time. The $O(k\log k)$ is for sorting the answers. The time for initial pattern search is $O(p)$. Putting all together with $\pi =\log\log n$, we obtain Theorem 3.  \qed
\end{proof}
\section{Space-Efficient Encoding of Our Index}

In this section, we derive a space-efficient index for  the relevance metric \emph{term-frequency}. The major contribution is that, instead of using $O(\log n)$ bits for an I-structure entry, we design some novel encodings so that each entry requires only $\log D+\log\pi+O(1)$ bits.  The GST will be replaced by a  compressed full text index $CSA$  of size $|CSA|$ bits~\cite{GGV03,fm,GroVit05,csanew} along with the tree encoding of GST in $4n+o(n)$ bits~\cite{tree}\footnote{Any $n$-node ordered tree can be represented in $2n + o(n)$ bits, such that if each node is labeled by its pre-order rank in the tree,
any of the following operations can be supported in constant time~\cite{tree}:
\emph{parent$(i)$}, which returns the parent of node $i$;  \emph{child$(i,q)$}, which returns the  $q$-th child of node $i$;
\emph{child-rank$(i)$}, which returns the number of siblings to the left of node $i$;
\emph{lca$(i,j)$}, which returns the lowest common ancestor of two nodes $i$ and $j$; and
\emph{lmost-leaf$(i)$/rmost-leaf$(i)$}, which returns the leftmost/rightmost leaf of node $i$.}.
Thus $locus(P)$ can be computed in $O(p)$ time by taking the LCA (lowest common ancestor) of leftmost and rightmost leaf in the suffix range of $P$.

A core component of our index is the document array $D_A$, where
$D_A[i]$ stores the id of document to which the $i$th smallest suffix in GST belongs to.
The $D_A$ can be maintained in $n\log D + O(\frac{n\log D}{\log\log D})$ bits and can answer the following queries in $O(\log\log D)$ time~\cite{Goly}.
(i) \emph{access$(i)$}:    returns $D_A[i]$;
(ii) \emph{rank$(d, i)$}:  returns the number of occurrences of document $d$  in $D_A [1... i]$;
(iii) \emph{select$(d, j)$}: is $-1$ if $j >|d|$, else $i$ where $D_A[i]=d$ and $rank(d,i)=j$.
Now we show how to use $D_A$ for efficient encoding and decoding of different components in an I-structure entry.

\emph{Term-frequency Encoding}: Given an I-structure entry with $origin = v$ and $doc = d$, the corresponding \emph{term-frequency} score is  exactly the number of occurrences of $d$ in $D_A[i...j]$, where $i$ and $j$ are the leftmost leaf and the rightmost leaf of $v$, respectively.  Thus, given the values $v$ and $d$, we can find $i$ and $j$ in constant time based on the tree encodings of the GST, and then compute \emph{term-frequency} in $O(\log\log D)$ time based on two rank queries on $D_A$.  Thus, we will discard the \emph{score} field completely for all I-structure entries, but keeping only the RMQ structure over it.

\emph{Origin Encoding}: Origin encoding is the most trickiest part, and is based on the following observation by Hon et. al~\cite{FOCS09}: for any document $d$ and for any node $v$ in GST, there is at most one ancestor of $v$ that contains an I-structure entry with $doc = d$ and $origin$ from a node in the subtree of $v$ (inclusively). We introduce two separate schemes for encoding \emph{origin} fields in \emph{near} and \emph{far} entries. This reduces the \emph{origin} array space from $O(n\log n)$ bits to $O(n)$ bits and decoding takes $O(\log\log D)$ time.

\emph{Encoding near entries:} Let $I_w$ be a regular I-structure (with only \emph{near} entries) associated with a node $w$ and let $w_q$ represents the pre-order rank of $q^{th}$ child of $w$.
 Then from the definition of I-structures, for a given document $d$, there exists at most one entry in $I_w$ with $doc =d$ and origin from the sub-tree of $w_q$ (inclusively).
  Thus, for a given document $d$ and an internal node $w$, an entry
 in $I_w$ can be associated to a unique child node $w_q$ of $w$ (where $w_q$ represent the $q$th child of $w$ from left, $1\leq q \leq degree(w)$, and pre-order rank of $w_q$ can be computed in constant time~\cite{tree}), such that \emph{origin} is in the subtree of $w_q$.
 Moreover, this \emph{origin} must be the node, closest to root, in the subtree of $w_q$ which has an N-structure entry for $d$. From the definition of N-structure, this \emph{origin} node must be the lowest common ancestor (LCA) of  the leaves corresponding to the first and last suffixes of $d$ in the subtree of $w_q$, which can be computed using the tree encoding of GST and a constant number of rank/select operations on  $D_A$ in total $O(\log\log D)$ time. Therefore, by maintaining the information about $w_q$ (\emph{origin-child} $= q$) for each I-structure entry, the corresponding \emph{origin} value can be decoded in $O(\log\log D)$ time.
 Thus, the \emph{origin} array can be replaced completely by the \emph{origin-child} array.  Recall that each node maintains the I-structure entries in sorted
order of the origins, so that the corresponding \emph{origin-child} array will be monotonic increasing.  In addition, the value of each entry is between $1$ and $degree(w)$, so that the array can be encoded using a bit vector of length $|I_w|+degree(w)$\footnote{A monotonic increasing sequence $S= 1333445$  can be encoded as $B=101100010010$ in $|B|(1+o(1))$ bits, where  $S[i] =rank_1(select_0(i))$ on $B$, and can be computed in constant time~\cite{RRR}.}.
The total size of the bit vectors associated with all nodes can be bounded by $\sum_{w \in GST}(|I_w|+degree(w))= O(n)$ bits.
The $O(n\log n)$ bits
predecessor search structure over \emph{origin} array is replaced by a structure of $o(n)$ bits space and $O(\log\log n)$ search time\footnote{Construct a new array by sampling every $\log^2 n$th element in the original array, and maintain predecessor search structure over it.
Now, when we perform the query, we can first query on this sampled structure to get an approximate answer, and the exact answer can be obtained by performing binary search on a smaller range of only $\log^2 n$ elements in the original array. The search time still remains $O(\log\log n)$.}.

 \emph{Encoding far entries:} In order to encode the \emph{origin} values in \emph{far} entries, we introduce the following notions. Let $w^*$ be a marked node, then another node $w_q^*$ is called its $q$th marked child, if $w_q^*$ is the $q$th smallest (in terms of pre-order rank) marked node with $w^*$ as its lowest marked ancestor. Given the pre-order rank of $w^*$, the pre-order rank of $w^*_q$ can be computed in constant time by maintaining an additional $O(n)$ bits structure.\footnote{Let GST$^*$ be a tree induced by the marked nodes in  GST, so that $w^*$ is the lowest marked ancestor of $w_q^*$ in GST if and only if the  node corresponding to $w^*$  in GST$^*$ (say, $w$) is the parent of node corresponding to $w_q^*$ (say $w_q$)  in GST$^*$.
Moreover,  $w_q^*$ is said to be the $q$th marked child of node $w^*$ in GST, if $w_q$ is the $q$th child of $q$ in GST$^*$.
Given the pre-order rank of any marked node in GST, its pre-order rank in GST$^*$ (and vice versa)
can be computed in constant time by maintaining an additional bit vectors of size $2n+o(n)$ which maintain the information if a node is marked or not.}
Let $I_{w^*}$ represents the combined I-structure (with only \emph{far} entries) associated with a marked node $w^*$. The origin value of any far entry in $I_{w^*}$ is always a node in the subtree of some marked child $w_q^*$ of $w^*$, and is always unique for a given $q$ and $doc=d$. Thus by maintaining the information about $w_q^*$ (\emph{origin-child$^*=q$}), we can decode the corresponding \emph{origin} value for a particular document $d$. i.e. \emph{origin} is the LCA of the leaves corresponding to the first and last suffix of $d$ in the sub-tree of $w_q^*$,
which can be computed using the tree encoding of GST and a constant number of rank/select operations on  $D_A$ in total $O(\log\log D)$ time. Now \emph{origin} array can be replaced by \emph{origin-child$^*$} array, which can be encoded in $\sum_{w^* \in GST^*}(|I_{w^*}|+degree(w^*))= O(n)$ bits (using the similar scheme for encoding \emph{origin-child} array for \emph{near} entries). The predecessor search structure is replaced by $o(n)$ bits sampled predecessor search structure.

\textbf{\emph{Query Answering}:} Query answering algorithm remains the same as that in our linear index, except the fact that decoding \emph{origin} and \emph{term-frequency} takes $O(\log\log D)$ time. Then the time complexities for the steps in Lemma 3 are as follows: Step (i) $O((\pi\log\log n+k)\log\log D)$, Step (ii) $O((\log \pi+k)(\log \pi+\log\log D))$ and Step (iii) $(((p/\pi)\log\log n+k)\log\log D)$. Since the \emph{term-frequencies} are positive integers $\leq n$, we shall use a y-fast trie~\cite{yfast} to get the sorted answer in $O(k\log\log n)$ time.
By choosing $\pi =\log^2\log n$, the query time can be bounded by $O(t_s(p)+p+\log^4\log n+k\log\log n)$, which gives the query time in Theorem 1. Here $t_s(p)$ is the time for initial pattern searching in $CSA$, and is $\Omega(p)$ for space-optimal CSA's~\cite{fm,csanew}.

\textbf{\emph{Space Analysis}:} The index consists of a full text index of $|CSA|$ bits, $D_A$ of $n\log D(1+o(1))$ bits, I-structures of total $2n(\log D+O(\log \pi)+O(1))$ bits,
tree encodings, RMQ structures and sampled predecessor search structures (together $O(n)$ bits). By choosing $\pi = \log^2\log n$, the index space can be bounded by $|CSA|+n\log D(3+o(1))+O(n\log\log\log n)$ bits. In order to obtain the space bounds in Theorem 1, we may categorize $D$ into the following two cases.

\begin{enumerate}
\item When $\log D/\log\log D > \log\log\log n$, the $O(n\log\log\log n)$ term can be absorbed in $o(n\log D)$. The space can be further reduced by $n\log D$ bits from the following observation that the \emph{term-frequency} is $1$ for those I-structure entries with \emph{origin = a leaf in GST}, and there are $n$ such entries. Therefore all such entries can be deleted and in case if such a document is within top-$k$, that can be reported using document listing. For that we shall use Muthukrishnan's chain array idea~\cite{muthu}. The chain array $C[1...n]$ is defined as follows: $C[i]= j$, where $j<i$ is the largest number with $D_A[i] = D_A[j]$ and can be simulated using $D_A$ as $j = select(D_A[i], rank(D_A[i], i)-1)$ in $O(\log\log D)$ time. Thus we do not maintain chain array, instead an $2n+o(n) =o(n\log D)$ bits RMQ structure~\cite{rmq} over it. 
Let $[L, R]$ be the suffix range of $P$ in the full text index, then document listing can be performed (in $O(\log\log D)$ time per document) by reporting all those documents $D_A[i]$ such that $L \leq i \leq R$ and $C[i] < L $ using repeated RMQ's. 
Although those documents with \emph{frequency} $ >1$ will get retrieved again (but only once),  it will not affect the overall time complexity. 

\item When $\log D/\log\log D \leq \log\log\log n$,
we shall use the index described in Theorem 4. Thus the space-query bounds will be $|CSA|+n\log D(1+o(1))$ bits and  $O(t_s(p)+\log\log n+k\log D \log^2\log D)=O(t_s(p)+k\log\log n)$ respectively.

\end{enumerate} 
By combining the above case, we get the result in Theorem 1. \qed

\BT
There exists an index of size $|CSA|+n\log D(1+o(1))$ bits with a query time of $O(t_s(p)+\log\log n+k\log D\log^2\log D)$ for retrieving top-$k$ documents with the highest term frequencies for a query pattern P of length p. 
\ET
\emph{Proof.} See Appendix C. 

\section{Saving More Space}
The most space-efficient version of our index (described in theorem 2) is proved in this section.
First, we give the following auxiliary lemma (see Appendix D for proof).  

\begin{lemma}
There exists an $O(n\log\sigma\log\log n)$ bits structure, which can answer access/rank/select queries on $D_A$ in $O(\log^2\log n)$ time, and can compute an entry $C[i]$ in the chain-array data structure (for document listing) in $O(\log\log n)$ time. 

\end{lemma}

To achieve space reduction, we categorize $D$ into the following cases:

\begin{enumerate}
\item  $\log D < (\log \sigma \log\log n)^{1+\epsilon/2} $: We shall use the index described in Theorem 4 and the query time will be $O(t_s(p)+k(\log \sigma \log\log n)^{1+\epsilon})$.

\item $\log D \geq (\log \sigma \log\log n)^{1+\epsilon/2}$: In this case $D_A$ is replaced by a structure described in Lemma 4, which makes the index space $n\log D(1+o(1))$ bits. Then by re-deriving the bounds with $\pi=\log^3\log n$, our query time will be $O(t_s(p)+\log^6\log n +k\log^2\log n)$.

The $O(k\log^2\log n)$ term can be further improved to $O(k\log\log n)$ from the following observation that, once we get the I-structure boundaries, we do not need any information about the $origin$ fields for further query processing. Thus the only value needed is the \emph{term-frequency}, which can be computed as follows:  a sampled document array $D_A^s$ is maintained, such that $D_A[i] = d$ is stored if and only if $(rank_{D_A}(d,i)) mod~\rho = 0$, for an integer $\rho = \Theta(\log D)$, else we store a NIL value, where $rank_{D_A}(d, j)$ is the number of occurrences of $d$ in $D_A[1...j]$.
Then $D_A^s$ can be maintained in $O(n\log D/\alpha) =O(n)$ bits and can compute an approximate rank.
That is $\rho~rank_{D_A^2}(d, j) \leq  rank_{D_A}(d, j) \leq \rho~rank_{D_A^2}(d, j)+\rho$. 
Thus associated with each I-structure entry, we shall store this error ($=\Theta(\log D)$), which is equal to actual \emph{term-frequency} minus approximate \emph{term-frequency} (computed using $D_A^s$). Thus by storing this error corresponding to each I-structure entry in total $O(n\log \rho) =O(n\log\log D) = o(n\log D)$ bits space, the \emph{term-frequency} can be obtained in $O(\log\log D)= O(\log\log n)$ time by first computing the approximate \emph{term-frequency} using $D_A^s$ and then by adding this stored value. Note that for the initial I-structure boundary searches, the origin decoding is performed using the structure in Lemma~4. Moreover, this structure can compute chain array values in $O(\log\log n)$ time, which can be used for document listing in $O(\log\log n)$ time per report (when the I-structure entries with \emph{term-frequency} $=1$ are deleted from the index, and later such a document is an answer for a query).

\end{enumerate}

By combining the above cases, we obtain an $|CSA|+n\log D(1+o(1))$ bits index with query time $O(t_s(p)+ k(\log \sigma \log\log n)^{1+\epsilon}+ \log^6\log n)$, which completes the proof of Theorem 2. \qed

\appendix

\section{Proof of Lemma 1}
In~\cite{FOCS09}, Hon et al.\ described an $O(t+k\log k)$-time algorithm for retrieving the $k$ largest numbers in the sorted order.
However, if sorted order is not necessary, the time can be improved to $O(p+k)$ based on the following result of Frederickson~\cite{Fred}:
The $k$th largest number from a set of numbers maintained in a binary
max heap $\Delta$ can be retrieved in $O(k)$ time by visiting $O(k)$ nodes in $\Delta$.
In order to solve our problem, we may consider a conceptual binary max
heap $\Delta$ as follows: Let $\Delta'$ denote the balanced binary subtree with $t$ leaves
that is located at the top part of $\Delta$ (with the same root).
Each of the $t-1$ internal nodes in $\Delta'$ holds the value $\infty$.
The $i$th leaf node $\ell_i$ in $\Delta'$ (for $i=1,2,...t$)
holds the value $A[M_i]$, which is the maximum element in the interval
$A[L_i..R_i]$. The values held by the nodes below $\ell_i$ will be defined recursively as follows:
For a node $\ell$ storing the maximum element $A[M]$ from the range $A[L..R]$,
its left child stores the maximum element in $A[L..(M -1)]$ and its right child
stores the maximum element in $A[(M+1)..R]$. Note that this is a conceptual heap which is built on the fly,
where the value associated with a node is computed in constant time based on the RMQ structures
only when needed. Therefore, we first find the $(t-1+k)$th largest element $X$ in this heap by visiting $O(t+k)$ nodes
(with $O(t+k)$ RMQ queries) using Frederickson's algorithm. Then, we obtain all those numbers
in $\Delta$ which are $\geq X$ in $O(t+k)$ time by a pre-order traversal of $\Delta$, such that if the value
associated with a node is $< X$, we do not check the nodes in its subtree. From those retrieved numbers,
we delete all the $\infty$s and then separate out the $k$ largest elements in $O(t+k)$ time.

\section{Proof of Lemma 2}
In order to answer the above query, we maintain $A$ in the form of a
\emph{wavelet tree}~\cite{GGV03}, which is an ordered
balanced binary tree of $n$ leaves, where each leaf is labeled with a
symbol in $\Pi$, and the leaves are sorted alphabetically from left to right.
Each internal node~$w_q$ represents an alphabet set~$\Pi_q$, and is
associated with a bit-vector $B_q$.
In particular, the alphabet set of the root is $\Pi$, and the alphabet
set of a leaf is the singleton set containing its corresponding symbol.
Each node partitions its alphabet set among the two children (almost) equally,
such that all symbols represented by the left child are lexicographically
(or numerically) smaller than those represented by the right child.

For a node $w_q$, let $A_q$ be a subsequence of $A$ by retaining only those symbols
that are in $\Pi_q$. Then $B_q$ is a bit-vector of length $|A_q|$,
such that $B_q[i]=0$ if $A_q[i]$ is a symbol represented by the left child of $w_q$, else $B_q[i]=1$.
Indeed, the subtree from $w_q$ itself forms a wavelet tree of $A_q$. To reduce the space requirement,
the array $A$ is not stored explicitly in the wavelet tree. Instead, we only store the bit-vectors $B_q$,
each of which is augmented with Raman et al.'s scheme~\cite{RRR} to support
constant-time rank/select operations. The total size of the bit-vectors and the augmented structures
in a particular level of the wavelet tree is $n(1+o(1))$ bits. We maintain an additional \emph{range
maximum query} (RMQ)~\cite{rmq} structure over the $score$ of all
elements of the sequence $A_q$ (in $O(|A_q|)$ bits). As there are $\log
\pi $ levels in the wavelet tree, the total space is $O(n\log \pi)$ bits.
Note that the value of any $A_q[i]$ for any given $w_q$ and $i$ can
be computed in $O(\log \pi)$ time by traversing $\log \pi$ levels in the
wavelet tree. Similarly given any range $[x'...x'']$ can be translated to
$w_q$ as $[x'_q..x''_q]$ in $O(\log \pi)$ time, where $A[x'_q..x''_q]$ is a subsequence
of $A[x'...x'']$ with only those elements in $\Pi_q$.

The desired $k$ highest scoring entries can be answered as follows:  Firstly the
given range $[y', y'']$ can be split into at most $2\log \pi$ disjoint
subranges, such that each subrange is represented by $\Pi_q$ associated
with some internal node $w_q$. All the numbers in the subsequence $A_q$
associated with such an internal node $w_q$ will satisfy
the condition $y' \leq A_q[i] \leq y''$. And for all such (at most $2\log
\pi$) $A_q$s, the range $[x',x'']$ can be translated into the corresponding range
$[x_q', x_q'']$ in $O(\log^2 \pi)$ time.
Now, we can apply Lemma 1 (where $t \leq 2\log \pi$) to solve the desired query.
However, retrieving a node value in the conceptual max heap (in the proof of Lemma 1)
requires us to compute the score of $A_q[i]$ for some $w_q$ and $i$ on the fly,
we shall do so by first finding the entry $A[i']$ that corresponds to $A_q[i]$, and then
retrieving the score of $A[i']$.  This takes $O(\log \pi + t_{score})$ time,
so that the total query time will be bounded
by $O(\log^2\pi+ (2\log \pi + k)(\log \pi + t_{score}))= O((\log \pi + k)(\log \pi+t_{score}))$.

\section{Proof of Theorem 4}
 A simple index can be derived based on the  succinct framework proposed by Hon et al.~\cite{FOCS09} and Gagie et al~\cite{gagie}, which consists of the compressed version of GST ($CSA$ and tree encoding) and the document array $D_A$  (of $n\log D +O(\frac{n\log D}{\log\log D})$ bits space with $rank/select/access$ capabilities in $O(\log\log D)$ time for any $d \in \D$~\cite{Goly}). Also, for a particular value $q$ to be defined
we group every $g= q\log D\log\log D$ leaves in the GST together (from left to right) and mark the lowest common ancestor (LCA) of all these leaves.  
Further, we mark the LCA of all pairs of marked nodes. Thus the number of marked nodes in GST can be bounded by $O(n/g)$~\cite{FOCS09}. For each marked node, we maintain the top-$q$ documents in its subtree explicitly, which takes $O(n/g \times q \log D)= O(n/\log\log D)$ bits.
We perform this marking and store the top-$q$ answers for $q=1,2,4,...$, which takes $O(\frac{n\log D}{\log\log D})$ bits of storage space. 
The total index space can thus be bounded by $|CSA|+ O(n)+ n\log D+O(\frac{n\log D}{\log\log D})= |CSA|+ n\log D(1+o(1))$ bits (assuming $D > \sqrt{\log\log n}$)

In order to retrieve top-$k$ answers corresponding to a suffix range, we first search for $P$ in GST and obtain its locus node $locus(P)$ in $O(p)$ time. 
Further we round the value of $k$ to the next highest power of $2$, say $q$. Now we search for a marked node $locus^*(P)$ (corresponding to this $q$), 
which is same as $locus(P)$ if $locus(P)$ is marked, else it is the highest marked descendent of $locus(P)$. Let $[L,R]$ be the suffix range of $locus(P)$ 
and $[L^*, R^*]$ be the suffix range of $locus^*(P)$, then the leaves corresponding to the ranges $[L,L^*-1]$ and $[R^*+1, R]$ are called \emph{fringe leaves}. 
It is easy to show that the number of fringe leaves is at most $2g$ (see~\cite{FOCS09}). Hence, in order to retrieve the top-$k$ answers, 
we first check the top-$q$ answers stored at $locus^*(P)$ (and compute their scores), and then retrieve the score of each of the $2g$ documents 
corresponding to the fringe leaves.  Recall that the score of a document $d$ is the frequency of $P$ in $d$, which can be computed in $O(\log\log D)$ time.
Thus, the total time can be bounded by $O(t_s(p) + (g+k)\log\log D)= O(t_s(p) + k\log D\log^2\log D)$.

Next, we find the top-$k$ answers from this candidate set of $2g+q <2g+2k$ documents. As there may be repetitions in the set, we first remove
the repetitions by scanning the set once (using an auxiliary bit vector of length $D$ to mark if we have already seen a document). 
After that, we find the document $d$ which has the $k$th highest frequency using $O(k+g)=O(k\log D\log\log D) $ time~\cite{median}. 
Finally, we isolate the top-$k$ answers in unsorted order based on the score of $d$, and sort them in $O(k\log k)=O(k\log D)$ time. 
If $D \leq \sqrt{\log\log n}$, we can retrieve the \emph{term-frequency} of all documents in $\D$ and trivially find the top-$k$ documents in $O(D\log D) = O(\log \log n)$ time. Putting all together, the over all query time can be bounded by $O(t_s(p)+\log\log n+k\log D\log^2\log D)$.

\section{Proof of Lemma 4}
Let $CSA$ be the compressed suffix array corresponding to the suffix array associated with $GST$.
Let $t_{sa}$ and $t_{\overline{sa}}$ denote the time for computing $SA[i]$ (starting position of $i$th smallest suffix of $T$) and the time 
for computing $SA^{-1}[j]$ (the rank of the $j$th suffix $T[j...n]$ among all suffixes of $T$), respectively. Hon et al.~\cite{FOCS09} showed that the above operations on $D_A$ can be simulated by an index of size $2|CSA|+o(n)$, and the best query time complexities are due to Belazzougui and Navarro in~\cite{Belazz}. 
We conclude the results in the following lemma. 
We maintain $CSA$ corresponding to $GST$ and the compressed suffix arrays $CSA_d$ (for $d=1, 2, 3,\ldots, D$)  corresponding to each individual document. Now, \emph{access$(i)$} can be obtained by returning $SA[i]$ in $CSA$. For \emph{select$(d, j)$}, we first compute the $j$th smallest suffix in $CSA_d$, and obtain the position $pos$ of this suffix within document $d$, based on which we can easily obtain the position $pos'$ of this suffix within the concatenated text of all documents. After that, we compute $SA^{-1}[pos']$ in $CSA$ as the desired answer for  \emph{select$(d, j)$}. 
By doing a binary search on $select$, \emph{rank$(d, i)$} can be obtained in  $O((t_{sa}+ t_{\overline{sa}})\log n)$ time. This time can be improved to $O((t_{sa}+ t_{\overline{sa}})\log\log n)$ as follows: At every $\log^2n$th leaf of each $CSA_d$, we explicitly maintain its corresponding position in $CSA$ and maintain a predecessor structure over it~\cite{yfast}. The size of this additional structure is $o(n)$ bits. Now, when we perform the query, we can first query on this predecessor structure to get an approximate answer, and the exact answer can be obtained by performing binary search on a smaller range of only $\log^2 n$ leaves. By choosing the $O(n\log\sigma\log\log n)$-bits space CSA by Grossi and Vitter~\cite{GroVit05}, where $t_{sa}$ and $t_{\overline{sa}}$ takes $O(\log\log n)$ time, we obtain the lemma.

An entry in chain array $C[i]= j$, if $j<i$ is the largest number with $D_A[i] = D_A[j]= (say~d)$ and is NIL if there is no such $j$. We shall use the following steps to compute $j$: using $SA[i]$ compute the starting position of lexicographically $i$th smallest suffix of the concatenated text and the corresponding $d$ value. Let this be the lexicographically $i_d$th smallest suffix  of $d$, then $(i_d-1)$th smallest suffix  of $d$ can be computed using an $SA_d$ and $SA_d^{-1}$ operations. Further we map this text position in $d$ back to the concatenated text and perform an $SA^{-1}$ operation on it to obtain $j$. The total time can be bounded by $O(\log\log n)$.

\end{document}